\documentclass[pdflatex,sn-mathphys-num]{sn-jnl}% Math and Physical Sciences Numbered Reference Style
\usepackage{graphicx}%
\usepackage{multirow}%
\usepackage{amsmath,amssymb,amsfonts}%
\usepackage{amsthm}%
\usepackage{mathrsfs}%
\usepackage[title]{appendix}%
\usepackage{xcolor}%
\usepackage{textcomp}%
\usepackage{manyfoot}%
\usepackage{booktabs}%
\usepackage{algorithm}%
\usepackage{algorithmicx}%
\usepackage{algpseudocode}%
\usepackage{listings}%
\def\a{\alpha}
\def\b{\beta}
\def\g{\gamma}

\def\={\overset{\bm .}{=}}
\def\e{\epsilon}

\def\sff{\mathrm{I\!I}}

\def\p{\partial}

\usepackage{ulem}

\usepackage[scr=rsfs]{mathalpha}

\theoremstyle{thmstyleone}%
\newtheorem{theorem}{Theorem}%  meant for continuous numbers
\newtheorem{proposition}[theorem]{Proposition}% 

\theoremstyle{thmstyletwo}%

\theoremstyle{thmstylethree}%
\newtheorem{definition}{Definition}%

\newtheoremstyle{manual}{}{}{\itshape}{}{\bfseries}{.}{ }{#1 #3}
\theoremstyle{manual}
\newtheorem*{manualthm}{Theorem}

\begin{document}

\title[Article Title]{Convex foliations and trapped submanifolds}

%%=============================================================%%
%% GivenName	-> \fnm{Joergen W.}
%% Particle	-> \spfx{van der} -> surname prefix
%% FamilyName	-> \sur{Ploeg}
%% Suffix	-> \sfx{IV}
%% \author*[1,2]{\fnm{Joergen W.} \spfx{van der} \sur{Ploeg} 
%%  \sfx{IV}}\email{iauthor@gmail.com}
%%=============================================================%%

\author[1,2]{\fnm{Gustavo} \sur{Dotti}}\email{gustavodotti@unc.edu.ar}

\affil[1]{\orgdiv{FAMAF}, \orgname{Universidad Nacional de C\'ordoba}, \orgaddress{\street{Ciudad Universitaria}, \city{C\'ordoba}, \postcode{(5000)}, \state{Cba}, \country{Argentina}}}

\affil[2]{\orgdiv{IFEG}, \orgname{CONICET}, 
\orgaddress{\street{Ciudad Universitaria}, \city{C\'ordoba}, \postcode{(5000)}, \state{Cba}, \country{Argentina}}}

%%==================================%%
%% Sample for unstructured abstract %%
%%==================================%%

\abstract{The conjecture that compact trapped submanifolds (CTMs) of any codimension greater than one 
 cannot intersect the domain 
of outer communications of a black hole is tested in symmetrically collapsing spacetimes of $n+1$ dimensions, 
$n \geq 2$, and on the entire Kerr-Newman sub-extreme family. The results provide  evidence 
to the idea that  CTMs of lower dimension, such as trapped loops,  
should be ragarded as black hole signatures.}


%%\pacs[JEL Classification]{D8, H51}

%%\pacs[MSC Classification]{35A01, 65L10, 65L12, 65L20, 65L70}

\maketitle

\tableofcontents

\section{Introduction}\label{sec1}
Let $(M,g)$ be an $n+1$ dimensional spacetime (Lorentzian, time oriented manifold), $S$ a submanifold.  Our 
conventions follows those  in \cite{Dotti2}:
the symbols $\perp$ and $\top$ denote respectively 
the normal and tangent component of  vectors  defined on $S$, $TS$ and $(TS)^\perp$ the tangent and normal bundles, 
$\mathfrak{X}(S)$ the set of (tangent) vector fields on $S$ and $\mathfrak{X}(S)^\perp$ the set of normal vector fields. 
The second fundamental form of $S \subset M$ is the  $(TS)^\perp$ valued 
symmetric $(0,2)$ tensor field on $S$ defined, for  $X, Y \in \mathfrak{X}(S)$, as
\begin{equation}\label{sffa}
\sff(X,Y)=-(\nabla_XY)^\perp.
\end{equation}
That this is indeed tensorial, that is, depends only on the values of $X$ and $Y$ at the point, and symmetric, 
is proved, e.g., in  \cite{onil}). 
The mean curvature vector field (MCVF) is the trace of $\sff$ over $TS$:
\begin{equation}
H ={\rm tr}|_S \sff.
\end{equation}
In components, in a -hopefully- self-explanatory notation:
 \begin{equation}\label{sffb}
 \begin{split}
\sff_{\a \b}^b &=   -(e_\a^a \nabla_a e_\b^b)^\perp.\\
H^b&=   -h^{\a \b} (e_\a^a \nabla_a e_\b^b)^\perp,
\end{split}
\end{equation}
where $\{ e^a_\a, \a=1,2,...,\text{dim} \, S\}$ is a local basis of vector fields on $S$ and $h^{\a \b}$ is the inverse of the induced metric $h_{\a \b}= g_{ab} 
e^a_\a e^b_\b$.  \\

Given any vector field  $\zeta \in \mathfrak{X}(M)$, define $S_t^\zeta$ as the flow of $S$ along $\zeta$, $V^\zeta(t)$ the volume of $S_t^\zeta$. 
If $S$ is compact one can show that \cite{onil,jost,li}
%$\p_t (V(S_t))|_{t=0} =: \dot V_\zeta$, which is 
\begin{equation}\label{Vdot}
\left. \frac{d V^\zeta(t)}{dt} \right|_{t=0} =\int_S H^b \zeta_b \, \e_o,
\end{equation}
$e_o$ the volume form of $S=S^\zeta_{t=0}$.  A  Compact Trapped subManifold (CTM) is  
a compact submanifold with future timelike MCVF. For these, the right hand side of \eqref{Vdot} is negative 
for future causal vector fields, so
CTMs can  equally be characterized as those compact submanifolds that initially shrink when flowed along 
future causal vector fields. 
A codimension two CTM is a { Compact Trapped Surface} (CTS), a one dimensional CTM is a 
Compact Trapped Loop (CTL). Compact trapped surfaces (in $3+1$ dimensions) 
were introduced by Roger Penrose in his fundamental paper \cite{Penrose}, 
where it is proved that, 
in a spacetime with a non-compact Cauchy surface and satisfying the null energy condition, they predict 
future null geodesic incompleteness. This result was extended by Galloway and Senovilla \cite{Gallo} 
to spacetimes of arbitrary dimensions and CTMs of higher codimensions:

\begin{theorem}[Galloway and Senovilla \cite{Gallo}]\label{thm1}
A spacetime $M^{n+1}$ is future null geodesically incomplete if
 it has a non compact Cauchy surface and contains a CTM $S$ of codim $\geq 2$ such that,  
along every orthogonal future null geodesic, 
\begin{equation}\label{cec}
R_{abcd} N^a N^c P^{bd} \geq 0
\end{equation}
holds from $S$ on, where $N^a$ is the tangent to the geodesic  and 
$P^a{}_b$ the parallel transport along the geodesic of the projector onto $T_pS$.
\end{theorem}

As shown in Proposition 1 in \cite{Gallo}, condition \eqref{cec} implies that long enough normal 
null geodesics reach an $S-$conjugate point, a condition used in the proof of Theroem \ref{thm1}. 
When $S$ is a surface (that is, codim $S=2$), the null energy condition implies  
\eqref{cec} and, together with  the above theorem, gives  the standard formulation Penrose's theorem. 
To prove this assertion, let $E_\a^a$ be the parallel 
transport of the $S$ tangent vectors $e^a_\a$ defined above, 
so that $P^{ab}=h^{\a \b}E_\a^a E_\b^b$. Complete 
 a frame along the geodesic by adding  $L^a$ future null, orthogonal 
to the $E_\a^a$ (as is $N^a$) and normalized as $L^aN_a=-1$, so that 
$g^{ab}=P^{ab}-L^a N^b -N^a L^b$, then 
\begin{equation}\label{nec}
R_{abcd} N^a N^c P^{bd} = R_{abcd} N^a N^c (P^{bd}+L^d N^b +N^d L^b)= R_{ac} N^a N^c,
\end{equation}
which is nonnegative if the NEC holds. Of course, the NEC is much more than what we need: 
Theorem \ref{thm1} 
only requires  \eqref{nec} to hold along the future half  of the  null geodesics orthogonal 
to the CS, and only for $N^a$ the geodesic tangent.  
Making this distinction is  crucial 
for higher codimension CTMs, because in this case 
   \eqref{cec} not only involves the Ricci part but also the Weyl piece of the curvature tensor 
   (so that  
  energy conditions alone 
cannot  define whether they are satisfied), 
and   is not expected to hold with such a generality,
 that is, at any spacetime point, for any null vector $N^a$ and orthogonal projector 
$P^{ab}$. 
This is proved in \cite{dd},  where it is also shown 
that, despite this,  
 Theorem \ref{thm1}  still finds  applications to predict singularities, thanks to the existence 
 of $k-$CTMs which are ``properly oriented'', that is,  so that \eqref{cec} holds. On the other hand, 
the standard proof  (as given,  e.g.  in \cite[Proposition~9.2.1]{hawk} or \cite[Proposition~ 2.2.2]{wald})  that CTSs do not 
intersect the domain of outer communications (DOC) 
of an asymptotically simple spacetime, cannot be generalized to $k-$CTMs of higher codimension, since the proof 
 is based
on the existence of conjugate points and would require 
that \eqref{cec} holds for \textit{every} $k-$CTM intersecting the DOC, not only for ``properly oriented'' ones. 
This is discussed in detail in \cite{dd}. \\
In \cite{Dotti2} and obstruction theorem for $k-$CTMs is given which can be used 
in two different ways: i) to rule out the existence of $k-$CTMs in spacetime regions foliated by 
$k-$future convex spacelike or null hypersurfaces and ii) to prove that no $k-$CTM can enter  a region 
foliated this way, and bounded by a leaf, 
from its future side.

This paper is part of a series aiming to emphasize the relevance of higher-codimension ($n+1-k > 2$)
 $k-$CTMs in spacetimes $(M^{n+1},g)$ with $n \geq 3$. The fact that a properly oriented $k-$CTM predicts the incompleteness of the future half 
of a
 null geodesic normal to it implies, by invoking weak cosmic censorship, that these CTMs should be regarded as black hole signatures. A natural question is then whether one can rule out the possibility that CTMs intersect the DOC, as happens with CTSs. In view of the comments above, the proof that works for CTSs does not generalize to higher-codimension CTMs, as they will not generically satisfy \eqref{cec}, which invalidates the proof of the existence of conjugate points. If the 
 conjecture that higher-codimension CTMs do not intersect the DOC  in general asymptotically simple spacetimes is true, 
 an alternative proof is needed. In this paper, this property is demonstrated for 
 the entire $(a,Q,M)$ Kerr-Newman family and in spherical collapse using the results in \cite{Dotti2}. 
 Section \ref{rev} offers a concise review of the results required from \cite{Dotti2}.

\section{Obstructions for CTMs}\label{rev}
\subsection{A Hierarchy of Convexity Conditions}

The following definition was introduced in \cite{Dotti2}:
\begin{definition} \label{dfn2} A  
spacelike or null hypersurface 
$\Sigma^n$ of a spacetime  $M^{n+1}$ 
is \emph{k-future convex} if, for any $p \in \Sigma$ and 
  any $k-$dimensional \textit{spacelike} subspace $W$ 
 of $V=T_p\Sigma^n$,
\begin{equation}\label{C1}
h^{\a \b}\,  e^a_\a e^b_\b \nabla_a N_b  \geq 0.
\end{equation}
Here $N^b$ is a vector field normal to $\Sigma$ and future pointing, 
$e^a_\a, \a=1,2,...,k$ a basis of $W$ and $h^{\a \b}$ 
the inverse metric matrix in this basis, so that 
the left side of \eqref{C1} is the $W-$trace 
of $(\nabla_a N_b)|_{W \otimes W}$.
\end{definition}

To get some intuition on this condition it is best to analyze separately 
the cases where $\Sigma^n$ is null or spacelike. \\

 \noindent
\uline{Spacelike hypersurfaces}: If $\Sigma^n$  is spacelike, $V=T_p \Sigma^n$ is an $n-$dimensional vector 
space with a positive definite induced metric $h_{ij}$, $i,j=1,2,...,n$, and  $k$ introduced  in Definition~\ref{dfn2} 
can take the  values $k=1, 2, 3, ..., n$. The 
set of $k-$dimensional subspaces of $V$ is the Grassmannian $\mathrm{Gr}(k,V)$, 
which is a compact manifold of dimension $k(n-k)$. Condition 
\eqref{C1} can be phrased as the non negativity of the function 
\begin{equation} \label{trs}
\mathrm{Gr}(k,V) \ni W \to {\rm tr}_W(K) \in \mathbb{R},
\end{equation} 
where $K_{ij}$ is the symmetric $(0,2)$ tensor on $V$ defined by  $(\nabla_a N_b)|_{V \otimes V}$. 
Since $K_{ij}$ is symmetric and $h_{ij}$ is positive definite, the $(1,1)$ tensor $h^{-1}K$ (that is,
 $h^{ik}K_{kj}$)  on $V$ 
admits a basis of eigenvectors with eigenvalues 
\begin{equation}\label{ev2}
\lambda_1  \leq \lambda_2 \leq  \cdots \leq \lambda_n.
\end{equation}
Now if $W$ is a $k<n$ dimensional subspace of $V$ with induced metric $h_W$, then, since the restriction 
$K|_W$ is symmetric and $h|_W$  positive definite, $(h_W)^{-1} (K|_W)$ is also diagonalizable. Its 
eingenvalues, however are   \textit{not} in general a subset of \eqref{ev2},  
unless $W$ is an invariant subspace 
of $h^{-1}K: V \to V$. It is shown in \cite{Dotti2} that the invariant subspaces of dimension $k$ are precisely 
the stationary points of the function \eqref{trs}, from where it follows that 
the global minimum of \eqref{trs} is $\sum_{A=1}^k \lambda_A$. This proves the following
\begin{proposition}\label{slc}
Assume the hypersurface $\Sigma^n$ is spacelike at $p$ with induced metric $h$ and let   $\lambda_1 \leq \lambda_2 ... \leq \lambda_n$ be the 
eigenvalues of $h^{-1}(\nabla N|_{T_p \Sigma \otimes T_p \Sigma}):
 T_p \Sigma \to T_p \Sigma$. For $k=1,2,...,n$,  $\Sigma$ is $k-$future convex at $p$ iff 
$\sum_{A=1}^k \lambda_A \geq 0$.
\end{proposition}

Utilizing this result, a clear connection can be made  between the  $k-$future convexity introduced in Definition \ref{dfn2} and
standard  concepts in extrinsic  Riemannian geometry (extrapolated 
to a Lorentzian background) and also established some useful properties of this condition:
\begin{enumerate}
\item $h^{-1}K: V \to V$ is the \textit{shape operator} at $p$, 
 its eigenvalues \eqref{ev2} are the \textit{principal curvatures} of $\Sigma^n$ at this point. 
\item \label{thisone} $k-$future convexity implies $k'-$future convexity for $k'>k$ \cite{Dotti2}.
\item  $1-$future convexity (the most stringent condition in this family) agrees with the concept of  local 
convexity, as defined e.g., in \cite{lima,bishop,bon},
 with an additional time orientation requirement, since $N^a$  future 
is assumed in \eqref{C1}, whereas $\pm N^a$ do the same in the standard settings. 
This explains the qualifier \textit{future} in ``$k-$future convexity''. 
\item $n-$future convexity  (the weakest one) is equivalent to the MCVF of 
$\Sigma$  being past pointing (see \cite{Dotti2}), that is,
$\Sigma^n$ being past trapped.
\item For $1<k<n$, the $k-$future convexity condition interpolates gradually 
between the previous two, as follows from 
item \ref{thisone} above, and does not seem to have an antecedent in the literature. 
\end{enumerate}

\noindent
\uline{Null hypersurfaces}: When $N^a$ is null we define 
$V$, $h_{ij}$ and $K_{ij}$ as in the spacelike case, but now 
$h_{ij}$ is degenerate, then $h^{-1}$ does not exist.  Also, the allowed $k$ values in Definition \ref{dfn2} 
are $k=1,2,...,n-2$ and the $k-$dimensional admissible  subspaces lie in the open 
subset $\widetilde{\mathrm{Gr}}(k,V) \subset \mathrm{Gr}(k,V)$ of  $k-$dimensional subspaces of $V$ 
\textit{with a positive definite induced metric}. 
Following \cite{Dotti2}, choose  a 
subspace of $V_o \subset V$ such that 
\begin{equation}\label{oplus}
V = V_o \oplus \text{span}\{ N \},
\end{equation}
(any $n-1$ dimensional subspace of $V$ not containing $N$ will do). 
We refer to  $V_o$ as a a \textit{section} of $V$ and define 
 $\pi: V \to V_o$ as the projector associated to the direct sum \eqref{oplus}. 
 It is easily checked (see \cite{Dotti2}) that, for any two vectors $u, v \in V$, 
 \begin{equation}
 \begin{split}
 h(u,v) &= h(\pi u, \pi v), \\
 K(u,v)&=K(\pi u, \pi v). 
 \end{split}
 \end{equation}
 This reduces the problem of finding the global minimum of 
\begin{equation} \label{trs2}
\widetilde{\mathrm{Gr}}(k,V) \ni W \to {\rm tr}_W(K) \in \mathbb{R},
\end{equation} 
to that of 
\begin{equation} \label{trs3}
\mathrm{Gr}(k,V_o) \ni U \to {\rm tr}_U(K) \in \mathbb{R},
\end{equation} 
 as 
${\rm tr}_W(K) = {\rm tr}_{\pi W}(K)$.  
This allows one to proceed as in the spacelike case:
\begin{proposition}\label{nc}
Assume the hypersurface $\Sigma^n$ is null at $p$.  Let $V_o$ be any section ($n-1$ dimensional spacelike subspace) of  $T_p \Sigma$, 
$h_o$ its induced metric 
and $\lambda_1 \leq \lambda_2 ... \leq \lambda_{n-1}$  the 
eigenvalues of $h_o^{-1} (\nabla N|_{V_o \otimes V_o}): V_o \to V_o$.
 For $k=1,2,...,n-1$,  $\Sigma^n$ is $k-$future convex at $p$ iff 
$\sum_{A=1}^k \lambda_A \geq 0$.
\end{proposition}

Note that the concept of $k-$future convexity adapts naturally to hypersurfaces with alternating null/spacelike 
sectors: $k$ ranges from $1$ to $n-1$ at points where the normal $N^a$ is null and from $1$  to $n$ 
at spacelike points. Propositions \ref{slc} and \ref{nc} are properly worded to be used in the 
case where $\Sigma^n$ changes character.\\

If the spacelike/null 
hypersurface $\Sigma^n$ in the above propositions is  embedded (we will stick to this case from now on), 
it  can  (locally) be regarded as a level set of 
a function $g:~M^{n+1} \to \mathbb{R}$. Choose the sign of $g$ such that 
$\nabla^a g$ is future. We say that a spacelike submanifold $S$ is tangent to $\Sigma^n$ at $p$ 
\textit{from its future side} if $g|_S$ has a local maximum at $p$. To explain the concept,  use 
the local maximum condition and take $O$ an 
open neighborhood of $p$ in $M$ such that $g(p) \geq g(q)$ for all $q \in O \cap S$. 
If $u^c \in T_pS$ then, as $p$ is a critical point of $g|_S$, $u^c \nabla_c g =0$. This 
proves the tangency condition  $T_pS  \subset T_p \Sigma^n$. 
Assume now that  $c(\tau)$   is a timelike curve  within $O$ from $\Sigma^n$ 
to $S$ (here $\tau \in [0,\tau_o]$ is proper time), then its tangent vector $v^a$ must be future. 
To prove this assertion, note that $\tfrac{d}{d\tau}g(c(\tau))= v^a \nabla_a g$. Since $\nabla_a g$ is future causal 
and $v^a$ timelike, $v^a \nabla_a g$ is either positive (then $v^a$ past) in the entire interval $0 \leq \tau \leq 
\tau_o$, or negative 
(then $v^a$ future) in this interval. However the first possibility 
is not allowed because $g(0)=g(p) \geq g(\tau_o)$. Thus the condition that $S$ be tangent to $\Sigma^n$ from 
its future side means there exist an open spacetime neighborhood of the tangency point $p$ 
such that, any timelike curve in $O$ from $\Sigma^n$ to $S$ must be future. Fore more 
details and subtleties 
related to this definition see \cite{Dotti2}.

\subsection{Obstructions for CTMs}

The main result in \cite{Dotti2} admits two equivalent formulations, given below as Theorem \ref{v1} 
and Theorem $4'$. The first version is better suited for calculations  in specific examples, whereas 
the second is more geometric. The main applications explored in \cite{Dotti2} 
were: i) to rule out the existence of CTMs 
of specific dimensions in open spacetime regions (e.g., the non-existence of CTSs within 
extremal black holes in the Kerr-Newman family); and ii) to find spacelike hypersurfaces that 
act as past boundaries of the region where $k-$CTMs are possible. \\
In this paper  the role of the theorem in proving that $k-$CTMs cannot intersect
the DOC is analyzed. To proceed,  the main result in \cite{Dotti2} are ateted in two different ways:

\begin{theorem}[\cite{Dotti2}]  \label{v1}
 Let $(M^{n+1},g_{ab})$ be  a spacetime, $g: M \to \mathbb{R}$ 
 a $C^2$ function and $Z_g^{(k)}$ an open  
 set where $\nabla^a g$ is future causal  and the trace of 
   the restriction of $\nabla_a \nabla_b g$ to spacelike $k-$dimensional subspaces of the tangent space of the $g-$level sets is 
 non-negative. 
 \begin{enumerate}[i)]
\item  If $S \subset M$ is a 
$k-$dimensional  spacelike submanifold and  $g|_S$ has a local maximum at $p \in Z_g^{(k)}$,
 then $S$ cannot satisfy the trapping condition at $p$.
 \item If $S$ is a $k-$dimensional CTM, then it is not possible that $S \subset 
Z_g^{(k)}$.
\end{enumerate}
 \end{theorem}
 
 \begin{proof} (Sketch)
 Introduce local coordinates $u^\a$  [$x^a$] for  $S$ [$M$] around $p$ and the associated local basis 
 $e^a_\a= \p x^a/\p u^\a$ of $TS$. Note that $T_pS \subset T_p \Sigma^n$, $\Sigma^n$ the level 
 set of $g$ through $p$, since, for every $\a$, 
  $e^a_\a \nabla_a g = \partial g(x(u))/\partial u^\a =0$ at $p$, where 
 $g|_S$ reaches a local maximum. 
 The induced metric $h_{\a \b} = g_{ab} e^a_\a e^b_\b$ has inverse $h^{\a \b}$. 
 Let $D_\a$ [$\nabla_a$] be the covariant derivative of $S$ [$M$], $\Delta_S g =h^{\a \b} D_\a D_\b g$  
 the $S$ Laplacian of $g|_S$. It can be shown (see \cite{Dotti2}) that the MCVF of $S$ at $p$ satisfies 
 \begin{equation}\label{c2}
 H^b \nabla_b g=h^{\a \b}\,  e^a_\a e^b_\b \nabla_a \nabla_b g - \Delta_S g.
\end{equation}
Since $p$ is a local maximum of $g|_S$, any coordinate 
Hessian 
$\p_\a \p_\b g$ of $g|_S$ at this point is negative semi-definite and agrees with 
$D_\a D_\b g=\p_\a \p_\b g - \Gamma_{\hspace{-1mm}S}^\g{}_{\a \b} \p_\g g$, then 
\begin{equation}\label{c3}
\Delta_S g=
h^{\a \b} (\p_\a \p_\b g - \Gamma_{\hspace{-1mm}S}^\g{}_{\a \b} \p_\g g) \Big|_p
 =h^{\a \b} \p_\a \p_\b g \Big|_p \leq 0.
\end{equation}
According to the hypothesis, $e^a_\a e^b_\b \nabla_a \nabla_b g \geq 0$, which, added to 
\eqref{c3} and \eqref{c2} gives $H^b \nabla_b g \geq 0$ at $p$. Since $\nabla^a g$ is future causal 
at $p$, this inequality implies that the MCVF $H^a$ cannot be future timelike at this point. This proves 1. To prove 2, 
we just note that, if $S$ is compact, $g|_S$ must reach a global, then local maximum.
 \end{proof}

 Since $\nabla^a g$ in the theory is future causal, then not zero (in our conventions the zero vector is not null), 
 the $g$ level sets are embedded hypersurfaces with normal covector $\nabla_a g$. The hypothesis 
 on $\nabla_a \nabla_b g$ can then be restated as the $g-$level sets being $k-$future convex (Definition 
 \ref{dfn2}). Finally, from the analysis in the ending paragraph of the previous section we recognize that 
 the local maximum condition translates into $S$ touching tangentially a level set hypersurface 
 from its future side. Since any embedded hypersurface can be locally regarded as a level set of a function 
 $g: M^{n+1} \to \mathbb{R}$, we can rephrase Theorem \ref{v1} as follows:

 \begin{manualthm}[4'] \label{v2} 
 Let $(M^{n+1},g_{ab})$ be  a spacetime.
 \begin{enumerate}[i)]
\item  If $\Sigma^n$ is a  $k-$future convex spacelike/null hypersurface and $S$  a spacelike 
$k-$dimensional submanifold tangent to  $\Sigma^n$  at $p$  from its future side, 
then  $S$ cannot satisfy the trapping condition at  $p$.
 \item If $Z^{(k)}$ is an open subset of $M$ 
 foliated by $k-$future convex spacelike/null hypersurfaces and  $S$ is a $k-$dimensional  CTM, then it is not possible that $S \subset 
Z^{(k)}$.
\end{enumerate}
\end{manualthm}

Part 1 of Theorem was used in \cite{Dotti2} to explain why there are no $k-$CTMs in certain spacetime 
regions (e.g., extremal black holes); part 2 to establish past spacelike barriers for $k-$CTMs, in particular, to 
show that the barrier found in \cite{Bengtsson:2010tj} 
for CTSs in $3+1$ spherical collapse spacetimes also works for TLs. In the following section,  part 1 
will be used with $\Sigma^n$ the event horizons of black holes.

\section{$k-$CTMs and the domain of outer communications}

In an asymptotically simple $3+1$ spacetime satisfying the NEC, a CTS $S$ 
cannot intersect the domain of outer communications 
$\mathcal{D}$, which is defined as the complement  $J^-(\mathscr{I^+})$ of the black hole region $\mathcal{B}$. 
The proof, given, e.g., as  Proposition 12.2.2 in \cite{wald}, uses the fact that, if $S \cap \mathcal{D} \not = 
\emptyset$ there is a future null geodesic orthogonal to $S$ and reaching $\mathcal{I}^+$, 
and that this geodesic has 
 no 
$S-$conjugate points. This contradicts the fact that future null geodesics orthogonal to $S$ 
  do have  $S-$conjugate points when $S$ is trapped (Proposition 9.3.9 in \cite{wald})

\subsection{Symmetrically collapsing spacetimes} \label{sc}

Consider a general, dynamic spacetime 
\begin{equation}\label{gcs}
ds^2 = -e^{2\beta(v,r)} f(v,r) dv^2 + 2e^{\beta(v,r)} dv dr + r^2 \gamma_{AB} d\theta^A d\theta^B,
\end{equation}
where  $\gamma_{AB}(\theta)d\theta^A d\theta^B$ is the metric of the unit $n-$sphere. Under suitable 
conditions  on $f$ and $\beta$ \eqref{gcs}  models 
spherically symmetric collapse with black hole formation  (see section VI in \cite{Bengtsson:2010tj}). If neither $\beta$ nor $f$ depend 
on $v$, \eqref{gcs} models a static spacetime in those regions where $f>0$. \\
 We choose the time orientation for which 
 the null vector field $-\p_r$ is future. We remind the reader 
that 
that apparent horizon (AH)  is the outermost connected component of the hypersurface  defined by 
$f(v,r)=0$ \cite{Bengtsson:2010tj}, and may have spacelike as well as timelike or null portions. The  event horizon (EH)  is a null, spherically symmetric  hypersurface 
located outside the AH, where $f>0$, and it is generated by outgoing radial null geodesics. These curves satisfy the equation 
\begin{equation}\label{ossng}
-e^{2\b}f \dot v+ 2 e^\b  \dot r =0,
\end{equation}
where a dot means derivative with respect to an affine parameter. Note that the EH is 
one of these null hypersurfaces. 
 The EH  and AH may coincide at some points or approach asymptotically towards the future. 
%as we approach a point in AH $\cap$ EH from the DOC, $f \to 0$ from the right.  
 Consider the spherically symmetric 
solutions of the eikonal equation 
\begin{equation} \label{eiko}
(\partial_r g) \left[2e^{-\beta} (\partial_v g) + f (\partial_r g)\right] = 0.
 \end{equation}
We are not interested in the  family  $g=g(v)$, whose level sets are null hypersurfaces generated by the ingoing   null vector field $-\p_r$,  
but in the solutions that make zero 
the term within square brackets, whose level sets are generated by the outgoing radial null geodesics \eqref{ossng}. We will assume that there is a regular solution of this equation in the DOC.
By regular we mean that $dg \neq 0$, so that its level sets are embedded null hypersurfaces. 
This forces $\p_r g \neq 0$ (otherwise $dg=0$, see \eqref{eiko}), so $\p_r g$ cannot change sign within the DOC. We lock 
our choice of sign for $g$ by adding the condition $\p_r g>0$, so that the null vector field $\nabla^a g$ is future. 
Some care is required if $\text{AH} \cap \text{EH} =Z \neq \emptyset$, 
as happens in the static case: suppose 
$(v_o,r_o,\theta, \phi)$ are the coordinates of a point in $Z$, then $\lim_{r \to r_o^+}f(r,v_o) = 0$ and 
so  $\lim_{r \searrow r_o} \p_r g =  \lim_{r \searrow r_o^+} [-2 e^\b (\p_v g)/f]$ is either a positive constant $c$ or $\infty$. 
In the first case the horizon is the level set $g=c$. In the second case, take 
 $r_1 >r_o$, then
 \begin{equation}
 g(v_o,r)=g(v_o,r_1)-\int_r^{r_1} \p_rg(v_o,r') \, dr' \to -\infty \; \text{ as } \;  r\searrow r_o,
 \end{equation}
 and $G=e^g$ is a new solution to the ``outgoing'' eikonal equation with the EH the level set $G=0$.
 In any case, a solution $g$ of the eikonal equation  in the DOC can be found for which  the EH, the past boundary of the black hole, is the 
 lowest level 
 set, as $g$  grows 
 in the DOC. In view of Theorem \ref{v1}.1, $k-$CTMs could not trespass the EH into the DOC if this foliation is
 $k-$future convex, as a local maximum of $g$ would be reached in the DOC. To check convexity at a point $p$ of 
 a $g$ level set $\Sigma$, 
 note that the $\p_{\theta^A}$ span a section $V_o$ of $T_p \Sigma$  and the restriction of 
 the covariant Hessian $\nabla_a \nabla_B g$ to $V_o \otimes V_o$ of a function 
 $g$ satisfying the outer eikonal equation simplifies to
 \begin{equation}
 \nabla_a \nabla_b g|_{(a=A,b=B)} = \frac{1}{2} r f \partial_r g \, \gamma_{AB},
 \end{equation}
which is  is positive definite since $r, f $ and $\p_r g$ are all  positive in the DOC.
We  conclude that the foliation by $g-$level sets is $1-$future convex and no $k-$CTMs can intersect the DOC. 
Note that the analysis of this example is based on the facts that  the spacetime is a warped product 
of a Lorentzian 2-manifold (the one with local coordinates $(v,r)$) with the unit $n-1$ sphere 
 $({\rm S}^{n-1},\gamma_{AB})$, which is Riemannian and on which 
 the isometry subgroup $SO(n)$ acts transitively. We could generalize it by replacing 
 $({\rm S}^{n-1},\gamma_{AB})$ by any homogeneous compact Riemannian manifold $(K, \gamma)$, that is,  
one on which  its isometry group $H$ (which is a subgroup  of the spacetime isometries) acts transitively 
The event horizon would then be $H-$invariant, and a level set  of an outgoing radial 
solution of the eikonal equation. This spacetime models an $H-$symmetric collapse, 
$H=SO(n)$ being a particular case. \\

\noindent
\textbf{Example:} Consider Vaidya spacetime,  given by  \eqref{gcs} with $\gamma$ the metric of the unit 2-sphere, $\beta=0$ and 
$f(v,r)=1-2m(v)/r$. We assume  $0 \leq m(v) \leq M$, $dm/dv \geq 0$, $m(v) \to 0$ as $v \to -\infty$ 
and $m(v) \to M$ as $v \to \infty$. The EH is  the null, spherically symmetric  hypersurface $v=V_E(r)$ 
defined by the equation $dV_E/dr=2/f(V_E(r),r)$ together with $\lim_{v \to \infty} V_E(r)=2M$. The AH 
is the  spherically symmetric hypersurface $r=2m(v)$; it is \textit{spacelike} where $dm/dv>0$, null in those sectors where $dm/dv=0$.
 The projection onto the $(r,v)$ 
plane of the $1-$future convex DOC slicing by $g$ level sets ($g$ defined below equation \eqref{eiko}) gives the set of outgoing radial null geodesics. 
The DOC slicing for the case 
\begin{equation} \label{mf}
m(v) = \tfrac{M}{2} (1 + \tanh(v/v_o)
\end{equation}
 is shown in Figure \ref{f} below. 

\begin{figure}[htbp]
    \centering
    \includegraphics[width=0.8\textwidth,trim=0cm 12cm 0cm 0cm, clip]{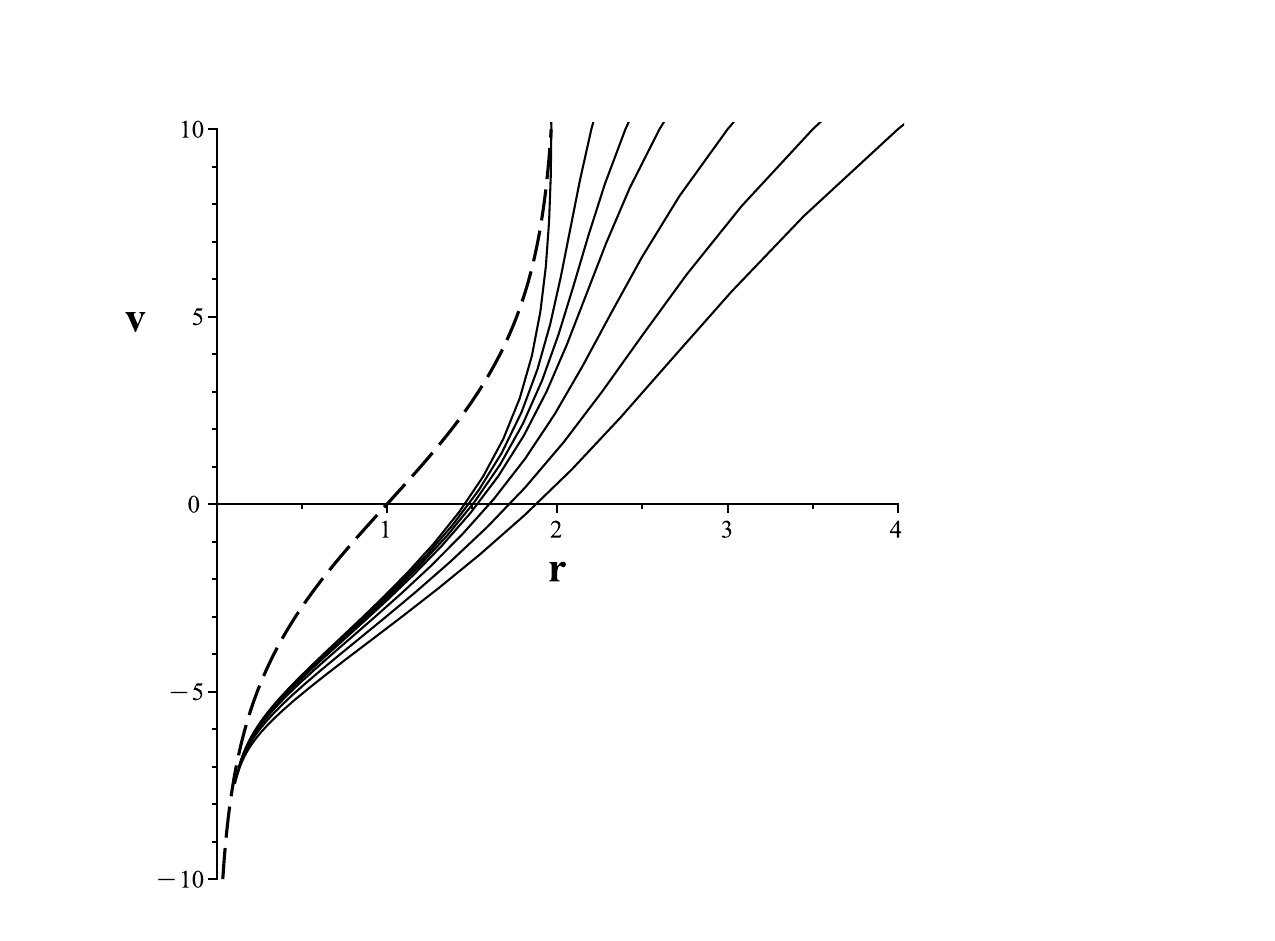}
    \caption{Seven leafs of the DOC slicing for the Vaidya spacetime with mass function \eqref{mf}  ($v_o=5$ and $M=1$)  
    are shown together with  the AH (dashed). Note the 
    gap between the AH and the EH.  Non spherically symmetric CTSs and TLs enter this gap up to a past barrier \cite{Dotti2,example,x,Bengtsson:2010tj}}
    \label{f}
\end{figure}

\subsection{Sub-extreme Kerr-Newman spacetime}

In advanced coordinates $(v,r,\theta,\phi)$, the metric of the Kerr-Newman  spacetime is 
\begin{multline}
ds^2 =  -\left(1- \frac{2Mr-e^2}{\rho^2} \right) dv^2+\rho^2 d\theta^2 + \left[ \frac{(r^2+a^2)^2- \Delta a^2 \sin^2\theta}{\rho^2} \right] \sin^2\theta  d\varphi^2\\
-\frac{2a(2Mr-e^2) \sin^2 \theta}{\rho^2} dv \, d\varphi + 2 \, dv \,dr -2 a \sin^2 \theta \ dr \, d\varphi. 
\end{multline}
Here $\Delta=r^2-2Mr+a^2+e^2$, $ \rho^2 = r^2+a^2 \cos^2 \theta$ and  $(\theta,\varphi)$ are the standard coordinates of $\rm{S}^2$. 
We may assume $a>0$ (otherwise switch $\varphi \to \varphi'=-\varphi$). 
The Cauchy ($r_-$) and event horizons $(r_+$) are the roots of $\Delta$, which are real by 
definition of ``sub-extreme'' and 
are located at $r_\pm = m \pm \sqrt{m^2-a^2-e^2}$; this can be inverted to
\begin{equation}\label{ax}
a=\sqrt{r_- r_+ -e^2}, \;\;\;\; m=\tfrac{1}{2} (r_-+ r_-).
\end{equation}
The eikonal equation $g^{ab} \nabla_a g \nabla_b g =0$ admits a family of solutions \cite{Carter:1968rr} 
\begin{equation} \label{eiko1}
g(v,r,\theta) = -v+ \int q_{s_1}(r) \; dr + s_2 \; a \sin \theta 
\end{equation}
where $s_1=\pm1$, $s_2=\pm 1$ are independent signs and 
\begin{equation} \label{eiko2}
q_{s_1}(r)= \frac{a^2+r^2}{\Delta}  + s_1  \frac{\sqrt{(r^2+2  m r -e^2)\, a^2+ r^4}}{\Delta},
\end{equation}
the content of the square root being  positive for all $r$ 
in the sub-extreme case. 
Our goal is to find a null foliation of the DOC (which is the set $r>r_+$) with the EH $r=r_+$ 
a level set. An inspection of \eqref{eiko1} may suggest that  this is impossible within this family of solutions due 
to its non-trivial dependence on $\theta$. There is, however, a loophole in this argument, as we now explain.  
If  $s_1=-1$ then  $q$ and $g$ are smooth across horizons; for $s_1=1$, instead, we find that 
$\lim_{r \searrow r_+} q_{(s_1=1)}(r) = \infty$, then,  if 
\begin{equation} \label{ssee}
g_o(v,r,\theta) = -v+ \int_{r^*}^r  q_{(s_1=1)}(r') \; dr' + s_2 \; a \sin \theta
\end{equation}
for some chosen  $r^*>r_+$, $ \lim_{r \searrow r_+} g_o(v,r,\theta) = -\infty$ and 
the solution 
$G= e^{g_o}$ of the eikonal equation has $\mathcal{H}$ as the level set $G=0$ with $G>0$ in the DOC, so this function  is perfectly suited   
 to Theorem \ref{v1}.1: any TL or CTS intersecting the DOC will necessarily reach a local maximum of $G$, 
which has $\nabla^a G$ null and future (as $-\p_r$ is future null and $\p_r G = e^{g_o} q_{(s_1=1)}(r)$, which is 
positive in the DOC). To make use of Theorem \ref{v1} we need to analyze the convexity of the $G$ level sets. 
To this end we find a section $V_o=\text{span}\{ e_1^a, e_2^a\}$ of  orthonormal vector fields orthogonal to $\nabla_a G$ 
and then analyze the positiveness of the symmetric matrix $G_{ij}=e_i^a e_j^b \nabla_a \nabla_b G$ 
(since $e_i^a e_j^b \nabla_a \nabla_b e^g = e^g e_i^a e_j^b \nabla_a \nabla_b g$ we analyze the positiveness 
of  $e_i^a e_j^b \nabla_a \nabla_b g$ instead). 
Starting with the non-orthonormal basis 
$\tilde e_1 = \p_{\varphi}$ and $\tilde e_2=a \cos(\theta) \p_v + \p_\theta$ we apply Gram-Schimdt and get an orthonormal basis 
with $e_1 \propto \tilde e_1$. We then analyze the trace and determinant of $G_{ij}$. The expressions are quite bulky. They 
were analyzed  numerically in the range $0<r_-<r_+$ , $e^2<r_-r_+$ after replacing \eqref{ax}; both the trace and determinant are 
positive in the DOC; an example is given in Figure \ref{f2}.
\begin{figure}[htbp]
    \centering
    \includegraphics[width=0.8\textwidth,trim=0cm 8cm 0cm 3cm, clip]{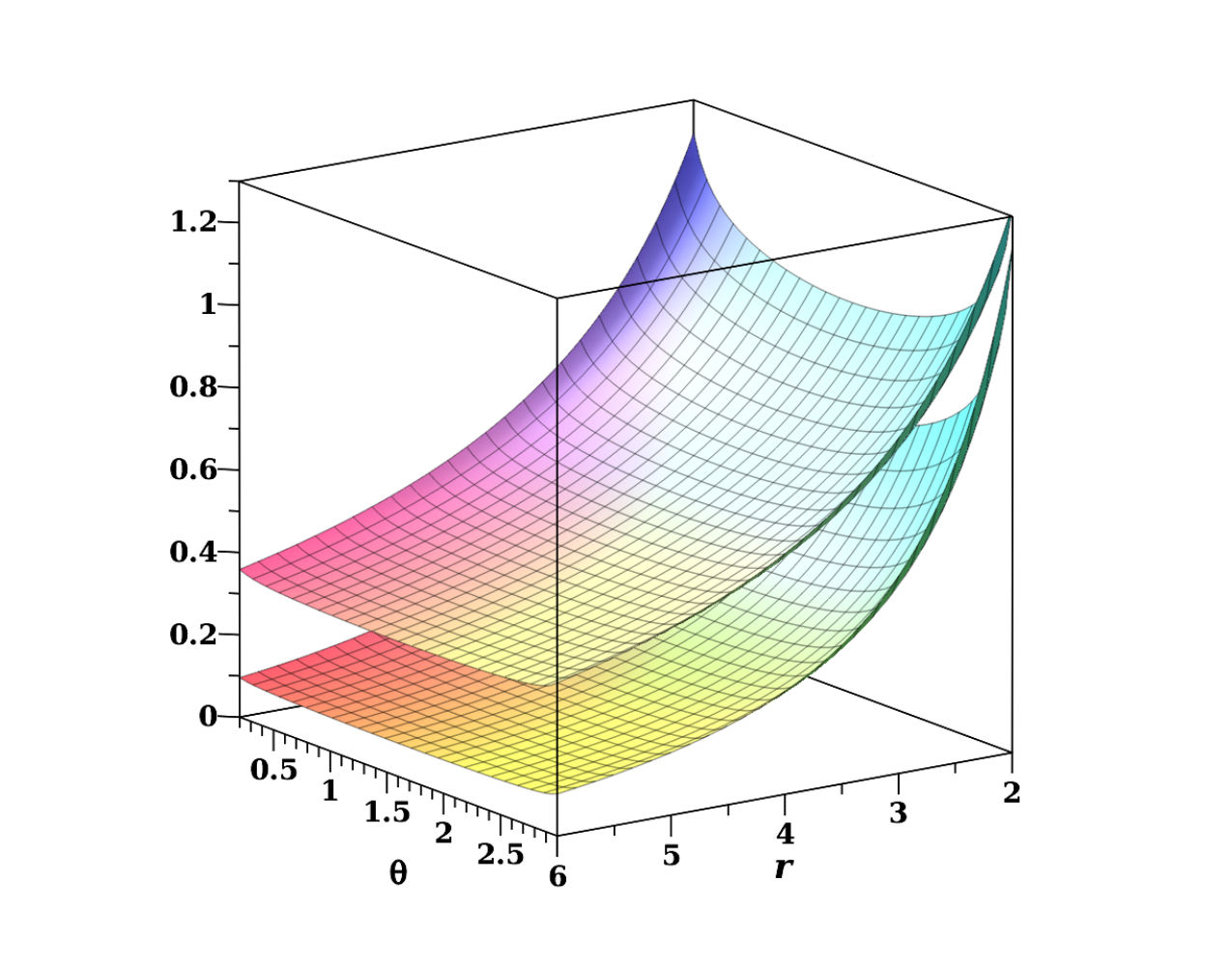}
    \caption{The trace (upper sheet) and (3x) the determinant of the matrix $e^a_i e^b_j (\nabla_a \nabla_b g)$ 
    for the solution \eqref{ssee} with $s_2=1$ of the Kerr-Newman eikonal equation. The plot corresponds to $r_-=1$, $r_+=2$ 
    and $e=1$. }
    \label{f2}
\end{figure}
The positiveness of both the trace and determinant implies that the $G$ foliation is $k-$future convex for 
$k=1,2$, then no TL or CTS can reach a local maximum of $G$; in particular, they can not be included 
and not even intersect the DOC. This statement is quite strong, as we now proceed to explain: 
 figure \ref{f3} shows schematically the EH  and the ergosphere 
 of the Kerr-Newmann spacetime in the non-trivial case $a \neq 0$ (note that the non rotating 
 case $a=0$ was already covered in section \ref{sc}). Five possibilities for a CTM are 
 considered, placed as in $A,B,C,D$ and $E$. A CTM like $A$ is possible: the exotic spheres 
 $(v=v_o, r=r_o)$ with $r_-<r_o<r_+$ are trapped, and their equators are TLs. $B$ could 
 not be a CTM, as there are no CTMs in stationary region. This is easily proved 
 using equation \eqref{Vdot} with $\zeta^a$ the timelike Killing field, for a CTM 
 $H^b$ is future timelike, then the integrand negative, whereas the volume 
 is invariant under the flow since the flow is an isometry, so we get a contradiction. 
 The non-trivial cases C, D and E are now ruled out due to the existence of 
 a $1-$future convex foliation of the DOC by $g$ level sets with the EH the minimum of $g$. 
 Any loop or compact surface like these would reach a local maximum of 
 $g$ within the DOC, and could not satisfy the trapping condition at this point. 
\begin{figure}[htbp]
    \centering
    \includegraphics[width=0.8\textwidth,trim=0cm 2cm 0cm 0cm, clip]{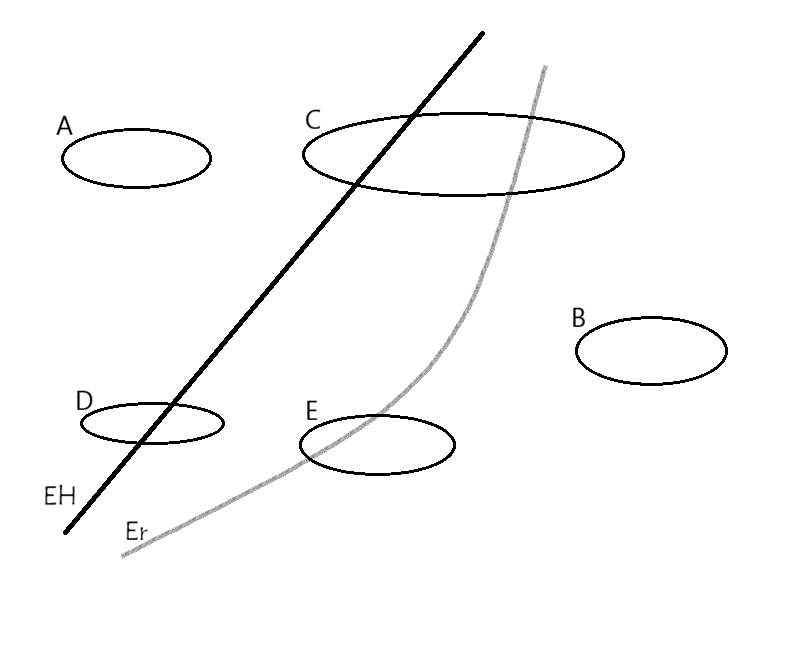}
    \caption{A schematic representation of a rotating sub-extreme black hole.  The figure 
    shows the EH (thick black straight line), the ergosphere (Er, thin curve) and 4 possibilities 
    of placing for  a CTM: A, B, C, D. Possibility A is allowed whereas B is not possible 
    because it lies  entirely within a stationary region. The cases C, D, E are ruled out by 
    the $1-$future convex foliation constructed in the text.
     }
    \label{f3}
\end{figure}
 
\section{Conclusions}
The canonical definition of the black hole region as $\mathcal{B}=M-J^-(\mathscr{I}^+)$, being  inherently global,   
is of limited conceptual and practical use. Both observationally and numerically we can only handle 
  a limited portion of the spacetime. Given that  Penrose's theorem (combined with the  weak cosmic censorship conjecture) implies 
   that closed trapped surfaces  (CTSs) can only exist  
  within  $\mathcal{B}$, the paradigm has shifted to 
  \textit{defining} $\mathcal{B}$ as the region of the spacetime containing CTSs. This region is limited 
  by stable marginally outer trapped surfaces which serve  as  (slicing dependent) 
  proxies of the event horizon  \cite{A1,A2,A3,Booth1,Booth3,Eardley,Pook-Kolb2,Booth2,newman}. \\
  Consider now higher odimension compact trapped submanifolds CTMs.  (In $3+1$ gravity, this brings trapped loops  into consideration.) 
   Although the hypothesis of the generalized singularity theorem by Galloway and Senovilla (reproduced here as 
   Theorem  \ref{thm1} in section \ref{sec1}) 
   cannot be satisfied by generic spacetimes 
 \cite{dd}, the theorem still works if we restrict ourselves to ``properly oriented'' $k-$CTMs \cite{dd}: 
 some of the null geodesics orthogonal to these CTMs are future incomplete and, again, 
 invoking weak cosmic censorship, we  conclude that higher (than two) codimension CTMs should also be regarded 
 as black hole phenomenology \cite{Dotti2,dd}. 
 Consequently, we posit the conjecture that no CTM can intersect the DOC $J^-(\mathscr{I}^+)$ of a canonical 
 black hole spacetime. 
    Tools to rule out the existence 
of CTMs  are given in \cite{Dotti}; then generalized in \cite{Dotti2} 
using the framework  of $k-$future convexity of spacelike/null hypersurfaces. It is 
proved that $k-$CTMS cannot be included (or intersect from its future side) spacetimes regions foliated by such hypersurfaces. 
In this work we have tested 
 the conjecture and demonstrated that it holds for every $k$ in the case of the symmetrically collapsing spacetimes in 
  \eqref{gcs}, which include spherical 
 collapse in arbitrary dimensions, as well as across  the entire sub-extreme Kerr-Newman family. The latter, contains non trivial  situations where the candidate 
 CTM could partially intersect the black hole and ergosphere regions.

\bmhead{Acknowledgements} Part of this work was presented at the  
\textit{Extremal Black Holes and the Third Law of Black Hole Thermodynamics ICERM Topical Workshop.} 
I thank the  ICERM for financial support.

\end{document}